\documentclass{llncs}

\usepackage{latexsym}
\usepackage{mathabx}
\usepackage{pst-all} 
\usepackage{pst-poly}
\newpsobject{showgrid}{psgrid}{subgriddiv=1,griddots=10,gridlabels=6pt}

\usepackage{algorithm}
\usepackage{algpseudocode}

\newcommand{\con}{\wedge} 
\newcommand{\dis}{\vee} 
\newcommand{\alw}{\Box} 
\newcommand{\imp}{\Rightarrow} 
\newcommand{\equ}{\Leftrightarrow} 
\newcommand{\som}{\Diamond} 

\title{A System for Deduction-based Formal Verification of Workflow-oriented Software Models}
\author{Rados{\l}aw Klimek}
\date{December 2013}
\institute{AGH University of Science and Technology,\\
           al.\ A.\ Mickiewicza 30, 30-059 Krakow, Poland\\
           \email{rklimek@agh.edu.pl}
           }

\newcommand\nonofoot[1]{%
   \begingroup
   \renewcommand\thefootnote{}\footnote{\kern-0.2ex#1}%
   \addtocounter{footnote}{-1}%
   \endgroup
}

\begin{document}
\framebox{
\begin{minipage}{.9\linewidth}
\begin{center}
\vspace{.5cm}
This document is a pre-print copy of the accepted article
for \emph{International Journal of Applied Mathematics and Computer Science},
2014,
vol.\ 24,
no.\ 4,
pp.\ 941--956,\\
available at\\
\texttt{https://www.amcs.uz.zgora.pl/?action=paper\&paper=802}
\vspace{.5cm}
\mbox{}\newline
The final version of the article is identifed by the following DOI:\\
\texttt{DOI:10.2478/amcs-2014-0069}
\vspace{.5cm}\\
\end{center}
\end{minipage}
}\\

\noindent\textbf{\Large BibTeX}
\vspace{-9cm}
\begin{verbatim}
@article{Klimek-2014-AMCS,
  author = {Klimek, Rados{\l}aw},
  title = {A System for Deduction-based Formal Verification of
           Workflow-oriented Software Models},
  journal = {International Journal of Applied Mathematics and
             Computer Science},
  year = {2014},
  volume = {24},
  number = {4},
  pages = {941--956},
  doi={10.2478/amcs-2014-0069},
  keywords = {formal verification, deductive reasoning,
              temporal logic, semantic tableaux,
              workflow patterns, logical primitives,
              generating logical specifications,
              business models, BPMN}
}
\end{verbatim}
\pagebreak
\pagestyle{plain}
\pagenumbering{arabic}
\setcounter{page}{1}
\maketitle

\begin{abstract}
The work concerns formal verification of workflow-oriented software models using deductive approach.
The formal correctness of a model's behaviour is considered.
Manually building logical specifications, which are considered as a set of temporal logic formulas,
seems to be the significant obstacle for an inexperienced user when applying the deductive approach.
A system, and its architecture, for the deduction-based verification of workflow-oriented models is proposed.
The process of inference is based on the semantic tableaux method
which has some advantages when compared to traditional deduction strategies.
The algorithm for an automatic generation of logical specifications is proposed.
The generation procedure is based on the predefined workflow patterns for BPMN,
which is a standard and dominant notation for the modeling of business processes.
The main idea for the approach is to consider patterns, defined in terms of temporal logic,
as a kind of (logical) primitives which enable
the transformation of models to temporal logic formulas constituting a logical specification.
Automation of the generation process is crucial for bridging the gap
between intuitiveness of the deductive reasoning and
the difficulty of its practical application in the case when
logical specifications are built manually.
This approach has gone some way towards supporting, hopefully enhancing our understanding of,
the deduction-based formal verification of workflow-oriented models.
\end{abstract}

\section{Introduction}
\label{sec:introduction}

Software modeling enables better understanding of domain problems and developed systems
through goal-oriented abstractions in all phases of a software development.
The software models require careful verification using
mature tools to make sure that the received software products are reliable.
Formal methods are intended to systematize and introduce a rigorous approach to
software modeling and development by providing precise and unambiguous description mechanisms.
A formal approach can be applied at any phase of the software-life cycle~\cite{Woodcock-etal-2009},
i.e.\ from requirements engineering to verification/validation as well as testing~\cite{Hierons-2009}.
A key issue in formal methods and software engineering is the correctness problem.
``Program testing can be used to show the presence of bugs, but never to show their absence''~\cite[Corollary]{Dijkstra-1972}.
Formal specification and formal verification are two important and closely related parts of the formal approach.
Formal \emph{specification} establishes fundamental system properties and invariants.
Formal \emph{verification} is the act of proving correctness of the system.
The importance of the formal approach increases and there are many examples of
its successful application, e.g.~\cite{Abrial-2007}.

This work concerns logical inference used for formal verification of software models
and practical possibilities of building tools for an appropriate verification procedure.
There are two fundamental and well-established approaches to formal verification of
systems~\cite{Clarke-Wing-etal-1996}.
The first one is algorithmically oriented and based on the state exploration
and the second one is logically oriented  and based on the deductive reasoning.
Now, the state exploration approach,
i.e.\ model checking~\cite{Clarke-etal-1999},
wins on points versus the deductive approach
due to the significant progress observed during recent years
in the field of model checking.
However, model checking is a kind of simulation for all reachable paths of computation
and constitutes an operational rather than an analytic approach.
On the other hand,
deductive reasoning plays an important role in the formal approach
as a ``top-down'' and sustainable way of thinking,
with reasoning moving from a more general facts to the more specific ones
to reach a certain conclusion which is logically valid.
Let us consider some arguments in favor of the deductive approach.
\begin{itemize}
  \item The main argument is the fact that deductive reasoning enables
        the analysis of infinite sequences of computations.
  \item Another argument is naturalness and common use of deductive reasoning in everyday life.
        It also dominates in scientific works.
  \item A kind of informal argument is an analogy between natural languages and logical approach,
        i.e.\ the application and knowledge of strict and formal grammatical rules,
        although not necessary, raises the quality and culture of statements in
        a natural language, while, by analogy, there is no doubt that
        applying strict logical rules for reasoning increases the quality of verification procedures and makes them more reliable.
\end{itemize}

Obtaining logical specifications which are considered as a~set of
temporal logic formulas $\{F_{1}, ..., F_{n}\}$ are
important and crucial issue for any deductive system.
When $n$ is large, which is not a rare situation even in the case of an average-size system,
then in practice it is not possible to build a logical specification manually and
therefore there is a need to automate this process.
However, software models could be organized into some predefined workflow patterns which
constitute a kind of primitives which
enable the transformation of software models to logical specifications.
Automation of this process enables the bridging of the gap between the naturalness of
the deductive verification and the difficulty of its practical application.
The lack of automation is a significant obstacle to the practical use of
logical inference for formal verification.
The choice of a~deductive system,
which is natural and intuitive enough for inexperienced users,
is another important aspect.
Although the work is not based on a particular method of reasoning,
the semantic tableaux method for temporal logic is selected since
it is intuitive and has some advantages to compare with
other deduction strategies.

Business models are considered in this work.
The significance of business models and their workflows increases
in the context of Service Oriented Architecture (SOA),
which is a paradigm that gained important attention within
the Information Technology (IT) and business communities.
All arguments mentioned in this Section are important for
research and constitute a challenge to the deductive approach.

\subsection{Motivation and contribution}
\label{sec:motivation-contribution}

The motivation of this work is the lack of tools for the automatic generation/extraction
of logical specifications, considered as sets of temporal logic formulas,
as well as the practical use of the deduction-based formal verification
for workflow-oriented models.
Business models expressed in BPMN (Business Process Modeling Notation),
a standard and the dominant notation for business processes,
are an important class of systems and are suitable for
the discussed method of deductive reasoning about system properties
and seem to be an intellectual challenge that software engineers are faced with
when they try to obtain trustworthy and reliable models.

The aim of this work is to provide a conceptual theoretical framework
supporting the deduction-based formal verification of workflow-oriented models.
The contribution is a complete deduction-based system, including its architecture and components,
which enables automated and formal verification of business models.
The main contribution is the algorithm for the generation of logical specifications
providing the method of extracting logical specifications from workflow models.
Theoretical possibilities of such an automation and the completeness issue for this process are discussed.
The application of a non-standard method for deduction which is the semantic tableaux method for
temporal logic in the area of business models is another contribution.
The proposed approach is characterized by the following advantages:
introducing predefined patterns as primitives to logical modeling,
and logical patterns once defined,
e.g.\ by a logician or a person with good skills in logic,
then widely used,
e.g.\ by analysts and developers with less skills in logic.

This work shows theoretical solutions to some problems as outlined above,
allowing for future preparation of workable practical solutions.
It also opens new research areas as shown in the last Section.

\subsection{Related works}

Workflow technologies are always important for the scientific world,
c.f.~\cite{Barker-VanHemert-2008},
providing a kind of glue for distributed services, for example,
service-oriented architectures which constitute a number of
loosely coupled and independent services,
to obtain more flexible than traditional and strictly coupled applications.
Thus, the importance of workflow technologies increases both for scientific and business domains.
\cite{Dehnert-Aalst-2004} presents
a kind of bridge between business process modeling and workflow specification.
The proposed methodology consists of some steps which are designed to provide and include a remedy for
intuitive description and informal languages.
A proliferation of business process management modeling languages is discussed in~\cite{Ko-etal-2009}.
Languages and notation are classified into groups of
execution, interchange, graphical standards, and diagnostics
providing identification and answer for some common misunderstandings,
and also discussing future trends.
The dominant language, and \emph{de facto} standard, for business process modeling
becomes BPMN (Business Process Modeling Notation),
see remarks at the beginning of Section~\ref{sec:analysis-verification}.
A formal semantics of a subset of BPMN using the process algebra CSP formalism is proposed
in~\cite{Wong-Gibbons-2011}.
Such formalism allows comparing BPMN models prepared by developers.
A pattern-based method expressing behavioural properties is considered in the work.
A translation into a bounded fragment of linear temporal logic is also presented.
In~\cite{Dijkman-etal-2008},
a mapping from BPMN to Petri nets is proposed to
obtain analysis techniques using existing Petri net-based tools, and to enable the static analysis of BPMN models.
In~\cite{Leuxner-etal-2010},
a formal model for workflows based on a transition system is presented and some algebraic properties are discussed.
A~meta-model for formal specification of functional requirements in business process models,
which is not well covered in literature, is proposed in~\cite{Frece-Juric-2012}.
Specific extensions to the BPMN semantic and diagram elements are introduced.
YAWL~\cite{Aalst-Hofstede-2005} is a workflow language supporting complex data transformations.
It is a graphical language but has a well-defined formal semantics
defined as a transition system
providing a firm basis for the formal analysis of real-world services.

Business models are also subject to formal verification.
\cite{Dury-etal-2007} discusses business workflows for
formal verification using model checking.
\cite{Eshuis-Wieringa-2004} addresses the issues of workflows
but they are specified in UML activity diagrams and the goal is to
translate diagrams into a~format that allows model checking.
Some aspects of workflows and
temporal logic are considered in~\cite{Brambilla-etal-2005}
but the formulas are created manually and
formal verification is not discussed very widely.
However, these considerations may constitute a kind of starting point for this work.
Another important direction of research is verifying
business processes using Petri nets~\cite{Aalst-2002}.
In~\cite{Zha-etal-2011},
a translation of workflows to Petri nets is proposed to perform analysis using existing tools.
An interesting direction of the analysis is $\pi$-calculus,
that enables efficient reasoning,
e.g.\ \cite{Ma-etal-2008},
and is designed for business processes and the BPEL language.
The paper~\cite{Bryans-Wei-2010} is another work that considers
an algorithmic translation from BPMN to the Event-B notation,
which based on the abstract machine notation, for system modeling and analysis.
\cite{Morimoto-2008} presents
a survey of formal verification for business processes. It discusses
automata,
model checking,
communicating sequential processes,
Petri nets, and
Markov networks.
All these issues are discussed in the context of business process management and web services.
In the general work by Shankar~\cite{Shankar-2009} automated deduction for verification is discussed.
There are discussed some important issues for symbolic logical reasoning,
e.g.\ satisfiability procedures, automated proof search, and variety of application
in the case of propositional and fragments of first-order logic.
However, even though the work contains a review of symbolic reasoning, modal and temporal logics are omitted.
\cite{Xu-etal-2012} discusses formal verification of workflows.
A special language is developed but algorithms refer only to propositional logic.
A deductive system for workflow models is proposed in~\cite{Rasmussen-Brown-2012}.
Even though it presents a solid mathematical framework and some deductive work is done,
the theoretical background is like Petri nets and not a formal logic.
In~\cite{Duan-Ma-2005},
a method, and a management system, for specification workflows by
temporal logic based workflow specification model is proposed.
\cite{Yu-Li-2007} proposes
a workflow and a linear temporal logic model.
It enables formal verification of workflows and is oriented on model checking.
\cite{Rao-etal-2008} proposes a process model of
a workflow management system for which specification of constraints are expressed in linear temporal logic.
Another paper that focuses on the constraints specification using linear temporal logic
is~\cite{Maggi-etal-2011}.
A translation of declarative workflow languages to linear temporal logic and finite automata
are considered in~\cite{Westergaard-2011}.
The improved algorithms for such a translation process are proposed.
In an interesting paper by Taibi and Ngo~\cite{Taibi-Ngo-2003} design patterns are discussed.
A simple language for pattern specification, combining first-order and temporal logic of actions,
is proposed.

However, all of the research themes mentioned above are different from
the approach presented in this paper which focuses on formal verification of
business processes using deductive-based reasoning with temporal logic.
While formal verification is discussed in some of the papers,
the application of temporal logic for this purpose is relatively rare.
Moreover, the deductive approach used for this domain is quite rare.

\subsection{Structure}

The rest of the paper is organized as follows.
Logical preliminaries which are temporal logic and logical inference using the semantic tableaux method
are discussed in Section~\ref{sec:preliminaries}.
Temporal logic is an established standard for the specification and verification of reactive systems
and the semantic tableaux method is a natural and valuable method of inference.
The deduction system and its architecture is proposed in Section~\ref{sec:deduction-system}.
The system enables formal verification of business models.
It consists of several software components, and some of them may be treated as interchangeable.
Workflow patterns are discussed in Section~\ref{sec:primitives}.
They are treated as (logical) primitives which allow to automate
the entire process of generating logical specifications.
The Algorithm for extracting logical specifications is proposed in Section~\ref{sec:generating-specifications}.
A general example of generating logical specifications is
presented in Section~\ref{sec:analysis-verification}.
The work is summarized and further research is discussed in Section~\ref{sec:conclusions}.

\section{Logical preliminaries}
\label{sec:preliminaries}

Formal logic is a symbolic language that supports
the reasoning process with statements to be evaluated as true or false.
There is a need for a rigorous and logic-based tool that
enables formal reasoning about software models.
Natural languages that do not belong to formal logic can be expressive but
they are very imprecise and ambiguous.
On the other hand,
formal languages, and such as formal logic, are not expressive
but they are precise and program properties expressed formally are
clearly and commonly understood.

\emph{Temporal Logic} TL which is a branch of symbolic logic that focusses on
statements whose evaluations depend on time flows,
i.e.\ it is a formal language which allows expression of temporal properties.
Temporal logic is a valuable formalism,
e.g.~\cite{Venema-2001,Wolter-Wooldridge-2011},
which has strong application in the area of software engineering for
the specification and verification of software models and reactive systems.
It is used for the system analysis where behaviors of events are of interest.
TL exists in many varieties, however, considerations in this paper are
limited to the \emph{Linear Temporal Logic} LTL,
i.e.\ logic for which the time structure is considered as a linear.
It means that each state has exactly one future.

The syntax of LTL logic is formulated over a countable set of
\emph{atomic formulas} $AP=\{ p, q, r, ...\}$ and
the set of \emph{temporal operators} ${\cal M}=\{ \som, \alw \}$.
Atomic formulas mean formulas with no propositional sub-structure,
or formulas with no sub-formulas,
or variables from propositional calculus.
Syntax rules allow the definition of syntactically correct,
or well-formed, temporal logic formulas.
\begin{definition}
\label{def:LTL-syntax}
A \emph{LTL formula} is a formula which is built using the following rules:
\begin{itemize}
\item if $p \in AP$ then $p$ is a LTL formula,
\item if $p$ and $q$ are formulas, then $\neg p$, $p \dis q$,
      $p \con q$, $p \imp q$, $p \equ q$ are LTL formulas
\item if $p$ is a formula, then $\boxcoasterisk\, p$, where $\boxcoasterisk \in {\cal M}$, is also a LTL formula.
\end{itemize}
\end{definition}
Thus, the whole \emph{LTL alphabet} consists of the following symbols:
$AP$, ${\cal M}$ and classical logic symbols like $\neg$, $\dis$, $\con$, etc.
It is relatively easy to introduce other symbols, e.g.\ parenthesis,
which are omitted here to simplify the presentation.
The ${\cal M}$ set consists two fundamental and unary temporal logic operators,
where $\som$ means ``sometime (or eventually) in the future''
and $\alw$ means ``always in the future''.
The operators are dual,
i.e.\ $\neg\som$ is, informally, equal to $\alw\neg$,
and $\som$ to $\neg\alw\neg$ and $\alw$ to $\neg\som\neg$.
The ${\cal M}$ set can be extended to other temporal logic operators.
Considerations in the work are focused on the LTL logic,
and particulary on \emph{Propositional Linear Temporal Logic} PLTL.
Propositions are statements that could affirm something about
members of a class, i.e.\ workflow activities considered in the work.
Thus, propositions $AP$ are used as atomic formulas in
Definitions~\ref{def:LTL-syntax} and~\ref{def:workflow-set},
as well as
atomic formulas in the predefined $P$ set in Fig.~\ref{fig:predefined-P} and~\ref{fig:predefined-P-Arbitrary-Cycles}.
However, notions introduced in these Figures are described in Section~\ref{sec:analysis-verification}.

The semantics of the LTL logic is traditionally defined using
the concept of \emph{Kripke structure}
which is considered as a graph, or path, whose nodes represent the reachable states
$w=s_{0},s_{1},s_{2},...$,
or in other words the reachable worlds,
and a labeling function which maps each node to
a set of atomic formulas $2^{AP}$ that are satisfied in a state.
A \emph{valuation} function $\nu(w(i)) \longrightarrow 2^{AP}$, where $i\geq 0$,
and $w(i)$ means the $i$-th element of the path $w$,
allows to define the \emph{satisfaction} $\models$ relation between
a path and a LTL formula, e.g.\
$w\models p$ iff $p\in w(0)$,
$w\models \neg p$ iff it is not $p\in w(0)$ and
$w\models \som p$ iff $p\in w(i)$, where $i\geq 0$, etc.
Theorems and laws of the LTL logic can be found in~\cite{Emerson-1990}.

Deductive reasoning is a kind of ``top-down'' way of thinking
that links premises and conclusions.
This is a typical and natural procedure in everyday life.
Logic and reasoning are cognitive skills.
Logical reasoning is the process of applying
sound mathematical procedures to given statements to arrive at conclusions.
Formal and logic-based inference enables reliable verification of desired properties.
There are some techniques, or proof procedures, which are systematic
methods producing proofs in some calculus, or provable, statements.
In other words, they are decision procedures for logic which
enables determining formula satisfiability.
There are some examples of deductive reasoning:
sequent calculi, resolution-based techniques or semantic tableaux.
The resolution technique is based on the observation that every logical formula can be transformed into a conjunctive normal form.
The interesting feature of the resolution method is that it has only one inference rule,
the resolution rule.
On the other hand,
the method can be employed to formulas (sub-formulas) in conjunctive normal form.
The essence of the procedure is to prove the validity of a sub-formula
by establishing that the negation of this sub-formula is unsatisfiable.
Another proof procedure is the semantic tableaux method which is based on the observation that
it is not possible for an argument to be true while the conclusion is false.
The essence of the procedure is finding counterexamples in branches of a tree after breaking down formulas.
Semantic tableaux are global, goal-oriented and ``backward'', while resolution is local and ``forward''.

Although the work is not based on any particular method of reasoning,
the method of semantic tableaux is presented in a more detailed way.
The method of \emph{semantic tableaux}, or \emph{truth tree}
is well known in classical logic but it can be applied
in modal logic~\cite{Agostino-etal-1999}.
\begin{figure}[htb]
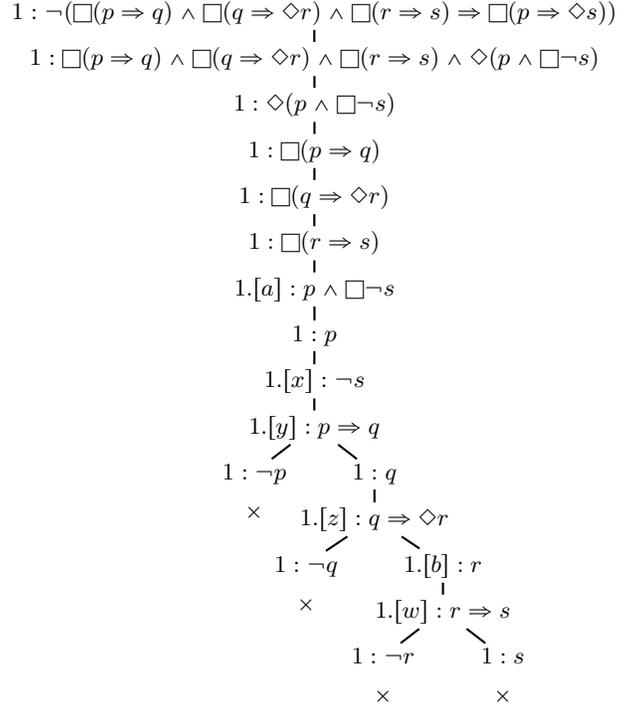

\centering
{\small
\pstree[levelsep=4.5ex,nodesep=2pt,treesep=25pt]
       {\TR{$1: \neg(\alw(p \imp q) \con \alw(q \imp \som r) \con \alw(r \imp s) \imp \alw(p \imp \som s))$}}{
          \pstree{\TR{$1: \alw(p \imp q) \con \alw(q \imp \som r) \con \alw(r \imp s) \con \som(p \con \alw\neg s)$}}{
          \pstree{\TR{$1: \som(p \con \alw\neg s)$}}{
          \pstree{\TR{$1: \alw(p \imp q)$}}{
          \pstree{\TR{$1: \alw(q \imp \som r)$}}{
          \pstree{\TR{$1: \alw(r \imp s)$}}{
          \pstree{\TR{$1.[a]: p \con \alw\neg s$}}{
          \pstree{\TR{$1: p$}}{
          \pstree{\TR{$1.[x]: \neg s$}}{
          \pstree{\TR{$1.[y]: p \imp q$}}{
            \pstree{\TR{$1: \neg p$}}{$\times$}
            \pstree{\TR{$1: q$}}{
              \pstree{\TR{$1.[z]: q \imp \som r$}}{
                \pstree{\TR{$1: \neg q$}}{$\times$}
                \pstree{\TR{$1.[b]: r$}}{
                \pstree{\TR{$1.[w]: r \imp s$}}{
                  \pstree{\TR{$1: \neg r$}}{$\times$}
                  \pstree{\TR{$1: s$}}{$\times$}
                }
                }
            }}
}}}}}}}}}}
}
\caption{The truth tree of the semantic tableaux method}
\label{fig:deduction-tree}
\end{figure}
It is a decision procedure for a formula satisfiability checking
and represents reasoning by contradiction,
i.e.\ \emph{reductio ad absurdum}.
The method is based on the formula decomposition
using predefined decomposition rules.
At each step of the well-defined procedure,
formulas become simpler as logical connectives are removed.
The tree is \emph{finished} if every (sub-)formula is decomposed and
every leaf contains an atomic formula or the negation of an atomic formula.
At the end of the decomposition procedure,
all branches of the received tree are searched for contradictions.
When the branch of the truth tree contains a contradiction,
it means that the branch is \emph{closed}.
When the branch of the truth tree does not contain a contradiction,
it means that the branch is \emph{open}.
When all branches are closed, it means that the tree is closed.
In the classical approach, starting from axioms,
longer and more complicated formulas are generated and derived.
Formulas are getting longer and longer with every step,
and only one of them will lead to the verified formula.
The method of semantic tableaux is characterized by the reverse strategy.
Though we start with a~long and complex formula,
it becomes less complex and shorter with every step of the decomposition procedure.
The open branches of the semantic tree provide
information about the source of an error, if one is found,
which is an advantage of this method.
\begin{example}{}
\label{example-semantic-tableaux}
The simple example of an inference tree for a temporal logic formula is shown in
Fig.~\ref{fig:deduction-tree}.
The formula of \emph{minimal temporal logic}~\cite{Chellas-1980,vanBenthem-1995} is considered.
The adopted decomposition procedure,
as well as labeling, refers to the first-order predicate calculus
and can be found in~\cite{Hahnle-1998}.
Each node contains a (sub-)formula which is either already decomposed,
or will be subjected to decomposition in the process of building a tree.
Each formula is preceded by a label referring to the current
world reference.
Label ``$1:$'' represents initial world in which a formula is true.
Label ``$1.(x)$'', where $x$ is a free variable,
represents all possible worlds that are consequent of the world~$1$.
On the other hand, label ``$1.[p]$'',
where $p$ is an atomic formula,
represents one of the possible worlds,
i.e.\ a successor of the world $1$,
where formula $p$ is true.
Let us note that all branches of the analyzed trees are closed ($\times$).
It means, there is no valuation that satisfies the root formula.
This consequently means that the formula before the negation,
i.e.\ $\alw(p \imp q) \con \alw(q \imp \som r) \con \alw(r \imp s) \imp \alw(p \imp \som s)$,
is always satisfied,
i.e.\ the formula is valid.
\end{example}

The semantic tableaux method can be treated as a \emph{decision procedure},
i.e.\ the algorithm that can produce the Yes-No answer as a response to some important questions.
Let $F$ be an examined formula and ${\cal T}$ is a truth tree built for a formula.
Then the following conclusions can be drawn.
\begin{corollary}{}
\label{th:decision-procedures}
The semantic tableaux method gives answers to the following questions related to the satisfiability problem:
\begin{itemize}
\item formula $F$ is not satisfied iff the finished ${\cal T}(F)$ is closed;
\item formula $F$ is satisfiable iff the finished ${\cal T}(F)$ is open;
\item formula $F$ is always  valid iff finished ${\cal T}(\neg F)$ is closed.
\end{itemize}
\end{corollary}
\begin{proof}
The semantic tableaux method is based on the systematic search for models that satisfy a formula.
To show that a formula is unsatisfiable, it needs to
show that all branches are closed.
Hence, if the tree is closed, it means there is no model that satisfies a formula.
To show that a formula is satisfiable, it needs to find one open branch.
If the tree is open, it means there exists a model that satisfies a formula.
If the tree for the negation of a formula is closed,
it means there is no model that satisfies a formula,
and as a result of the fact that this is a proving by contradiction,
it leads to the conclusion that the initial formula is always valid.
\end{proof}

\section{Deduction system}
\label{sec:deduction-system}

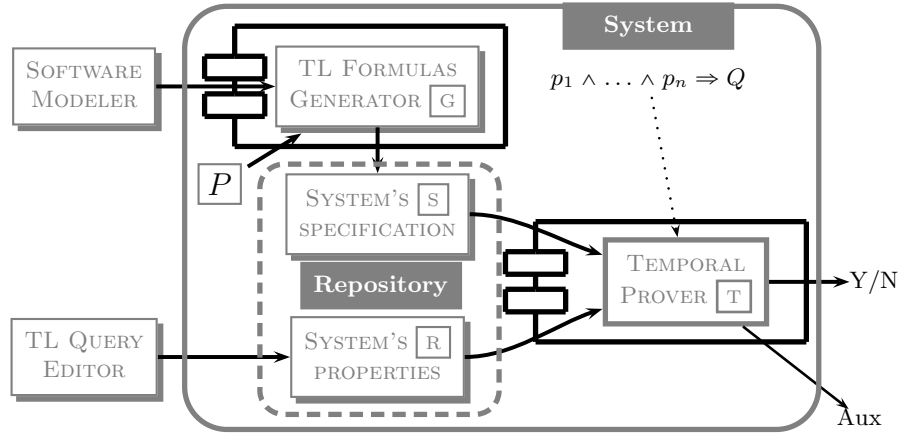
\begin{figure*}[htb]
\centering
\begin{pspicture}(11.5,6) 
\psset{framearc=0}
\psset{shadow=true,shadowcolor=gray}
\psset{linecolor=gray}
\rput(1.0,4.6){\rnode{r1}{\psframebox
                      {\begin{tabular}{c}
                            \textcolor{gray}{\textsc{Software}}\\
                            \textcolor{gray}{\textsc{Modeler}}
                      \end{tabular}}}}
\rput(4.9,4.6){\rnode{r2}{\psframebox
                      {\begin{tabular}{c}
                            \textcolor{gray}{\textsc{TL Formulas}}\\
                            \textcolor{gray}{\textsc{Generator} \psframebox[boxsep=true,shadow=false]{\scriptsize G}}
                      \end{tabular}}}}
\rput(4.9,2.9){\rnode{r5}{\psframebox
                      {\begin{tabular}{c}
                            \textcolor{gray}{\textsc{System's} \psframebox[boxsep=true,shadow=false]{\scriptsize S}}\\
                            \textcolor{gray}{\textsc{specification}}
                      \end{tabular}}}}
\rput(4.9,1){\rnode{r6}{\psframebox
                      {\begin{tabular}{c}
                            \textcolor{gray}{\textsc{System's} \psframebox[boxsep=true,shadow=false]{\scriptsize R}}\\
                            \textcolor{gray}{\textsc{properties}}
                      \end{tabular}}}}
\rput(1.0,1){\rnode{r7}{\psframebox
                      {\begin{tabular}{c}
                            \textcolor{gray}{\textsc{TL Query}}\\
                            \textcolor{gray}{\textsc{Editor}}
                      \end{tabular}}}}
\rput(9.0,2.0){\psset{linecolor=gray,shadow=false,linewidth=2pt}
               \rnode{r8}{\psframebox
                      {\begin{tabular}{c}
                            \textcolor{gray}{\textsc{Temporal}}\\
                            \textcolor{gray}{\textsc{Prover} \psframebox[linewidth=1pt,boxsep=true,shadow=false]{\scriptsize T}}
                      \end{tabular}}}}
\rput(2.8,3.3){\rnode{r2p}{\psframebox[boxsep=true,shadow=false]{\large $P$}}}
\rput(8.5,4.7){\rnode{r9}{$p_{1}\con \ldots \con p_{n} \imp Q$}}
\rput(11.5,2){\rnode{r10}{Y/N}}
\rput(11.3,.2){\rnode{r10aux}{{\small Aux}}}

\psset{linewidth=1.5pt}
\psset{shadow=false}
\psset{linecolor=black}
\ncline[angleA=90,angleB=0]{->}{r1}{r2}
\ncline[nodesep=.08cm]{->}{r2p}{r2}
\ncline[angleA=0,angleB=180]{->}{r2}{r5}
\ncline[angleA=270,angleB=90]{->}{r7}{r6}
\nccurve[angleA=0,angleB=162]{->}{r5}{r8}
\nccurve[angleA=0,angleB=198]{->}{r6}{r8}
\ncline[linestyle=dotted,linewidth=1pt]{->}{r9}{r8}
\ncline[angleA=90,angleB=0]{->}{r8}{r10}
\ncline[angleA=290,angleB=180,linewidth=1.0pt]{->}{r8}{r10aux}
\newcommand{\UMLkomponent}
           {\psline{-}(0,0)(3.6,0) \psline{-}(3.6,0)(3.6,1.6) \psline{-}(3.6,1.6)(0,1.6)
            \psline{-}(0,1.6)(0,1.2)
            \psline{-}(0,1.2)(-.4,1.2) \psline{-}(-.4,1.2)(-.4,.9)
            \psline{-}(-.4,.9)(.4,.9) \psline{-}(.4,.9)(.4,1.2) \psline{-}(.4,1.2)(0,1.2)
            \psline{-}(0,.9)(0,.7)
            \psline{-}(0,.7)(-.4,.7) \psline{-}(-.4,.7)(-.4,.4)
            \psline{-}(-.4,.4)(.4,.4) \psline{-}(.4,.4)(.4,.7) \psline{-}(.4,.7)(0,.7)
            \psline{-}(0,0)(0,.4)
           }
\psset{linewidth=2pt,linecolor=black}
\rput(3.0,3.8){\UMLkomponent}
\rput(7.0,1.2){\UMLkomponent}

\rput(2.3,0.0){\psframe[linewidth=2pt,framearc=.2,linecolor=gray](8.5,5.7)
\rput(1.0,0.2){\psframe[linewidth=2pt,framearc=.2,linecolor=gray,linestyle=dashed](3.25,3.40)}
\rput(2.65,1.95){\psframebox[fillstyle=solid,fillcolor=gray,linecolor=gray,shadow=false]{\textcolor{white}{\textbf{\textsc{Repository}}}}}
\rput(6.2,5.4){\psframebox[fillstyle=solid,fillcolor=gray,linecolor=gray,shadow=false]{\textcolor{white}{\textbf{\textsc{\quad System\quad}}}}}}
\end{pspicture}
\caption{An architecture of the deduction-based verification system}
\label{fig:deduction-system}
\end{figure*}
The architecture of the proposed inference system is presented and discussed below.
The system consist of some independent components and is shown
in Fig.~\ref{fig:deduction-system}.
The simpler version of the system is shown in~\cite{Klimek-2012-icaart}.
The system has two inputs.
The data stream with software models to be analyzed is the first input.
The approach is based on organizing models into
predefined patterns whose temporal properties are once defined,
e.g.\ by a person with good skills in logic,
then widely used,
e.g.\ by analysts with less skills in logic.
The second input is the analyzed property/properties expressed in terms of temporal logic formulas.
The easiest way to introduce such formulas is to use a plain text editor and to build them manually.
Such formula, or formulas, are identified by an analyst and
describe the expected/desired properties for the investigated software model.
Although specifying properties still requires knowledge of temporal logic,
but on the other hand formulas for properties are usually much easier to formulate.
The output of the whole deductive system is the ``Yes/No'' answer
in response to a new verified property.
The whole system can be synthesized informally as
$System(Model,Property) \longrightarrow Y/N$.
Such a process of inference can be performed many times in response
to any new formulas describing the desired and analyzed property.
There is another output that is called ``Aux''.
This is a point which enables outputting the auxiliary information
depending on the particular method of inference,
e.g.\ open branches in the case of the semantic tableaux method.

The proposed system is based on deductive reasoning and
enables examining whether a formula logically ``follows'' from some statements (formulas).
\begin{definition}
Let ${\cal U}$ is a set of formulas and $G$ is a formula.
If for every model of ${\cal U}$,
the formula $G$ is satisfied, i.e.\ the logical value of the formula is equal to the truth,
then $G$ is a \emph{logical consequence}, i.e.\ ${\cal U} \models G$.
\end{definition}
\begin{theorem}{Deduction Theorem}
Let ${\cal U} = \{F_{1},F_{2},...,F_{n}\}$.
${\cal U} \models G$ \quad iff \quad
$\models F \imp G$,
where $F_{1} \con ... \con F_{n} \equiv F$.
\end{theorem}
This is a well-known statement about equivalence of logical consequence and logical implication.
The proof could be found in~\cite{Kleene-1952}.
Summing up,
the examined formula $G$ is a logical consequence of $F$ iff statement $F \imp G$ is a \emph{tautology},
i.e.\ a statement that is always true.
It provides the important relation between the notions of logical consequence and validity.
The conjunction of all premises leads to the conclusion of the examined formula validity.

The system works automatically and consists of some important elements.
Some of them can be treated as a software components/plugins,
i.e.\ they are designed to work as a part of a larger system
introducing a specific feature,
and can be exchanged for one another with similar features if necessary.
The first component \psframebox[boxsep=true,shadow=false]{\scriptsize G}
generates logical specifications,
i.e.\ it performs mapping from software models to logical specifications.
This process depends also on the predefined workflow property set~$P$
which describes the temporal properties for every workflow
and is discussed in the next sections and shown in Fig.~\ref{fig:predefined-P} and~\ref{fig:predefined-P-Arbitrary-Cycles}.
A logical specification is a~set of a (usually) large number of temporal logic formulas
and is defined in Section~\ref{sec:generating-specifications}.
The generation of formulas is performed automatically by extracting logical specifications from
workflow patterns contained in a workflow model.
Formulas considered as a logical specification are collected
in the \psframebox[boxsep=true,shadow=false]{\scriptsize S} module
(data warehouse, i.e.\ file or database) that stores the specification of a system.
It can be treated as a~conjunction of formulas $p_{1}\con \ldots \con p_{n} = S$,
where $p_{i}$ is a specification formula generated during the extraction process.
The \psframebox[boxsep=true,shadow=false]{\scriptsize R} module provides
the desired and examined properties of the system,
as described above, which are expressed in temporal logic.
Both the specification of a system and the examined properties constitute
an input to the \psframebox[boxsep=true,shadow=false]{\scriptsize T} component,
i.e.\ \emph{Temporal (Logic) Prover},
which enables the automated reasoning in temporal logic.
The input for this component is usually formed in the form of the formula $S \imp Q$,
or, more precisely:
\begin{eqnarray}
p_{1}\con \ldots \con p_{n} \imp Q \label{initial-formula}
\end{eqnarray}
Due to the fact that the semantic tableaux method is an indirect proof,
then after the negation of Formula~\ref{initial-formula},
it is placed at the root of the inference tree
and decomposed using well-defined rules of the semantic tableaux method.
If the inference tree is closed, this means that
the initial Formula~\ref{initial-formula} is true.
The output of the \psframebox[boxsep=true,shadow=false]{\scriptsize T} component,
and therefore also the output of the whole deductive system,
is the answer Yes/No in response to any new verified property.

The whole verification procedure can be summarized as follows:
\begin{enumerate}
\item automatic generation of system specifications
      (the \psframebox[shadow=false,boxsep=true]{\scriptsize G} component),
      and then stored in the \psframebox[boxsep=true,shadow=false]{\scriptsize S} module;
\item introduction an examined property of a model
      (the \psframebox[boxsep=true,shadow=false]{\scriptsize R} module)
      as a temporal logic formula (formulas);
\item the automatic inference using semantic tableaux
      (the \psframebox[boxsep=true,shadow=false]{\scriptsize T} component)
      for the whole complex Formula~\ref{initial-formula}.
\end{enumerate}
Steps from~1 to~3, in whole or chosen, may be processed many times,
whenever the specification of the model is changed (step~1) or
there is a need for a new inference due to
the revised system's specification (steps~2 or~3).

\section{Workflows as primitives}
\label{sec:primitives}

Workflows considered as primitives are discussed in this Section.
\emph{Primitives} are primary or basic units not developed from anything else.
In the case of workflows,
they can be recognized as a low-level objects that lead to higher-level constructions.
In the case of logic,
they can be recognized as a not derived logical elements that lead to more complex logical specifications.
A combination of these two primitives is presented below.

Workflows play an important role in computer science and software engineering.
Broadly speaking,
the \emph{workflow} is a series of tasks, or procedural steps, or activities,
requiring an input and producing an output,
i.e.\ some added value to the whole activity.
In other words, the workflow enables observable progress of the work done by
a person, computer system, or company.
There are many examples of workflows and their notations
that influence computer science,
and one of them is business models, discussed in Section~\ref{sec:analysis-verification},
or activity diagrams of the UML language~\cite{Booch-Rumbaugh-Jacobson-1999,Pender-2003}.
The important feature of workflows is the fact that
they are focused on processes rather than documents.
This feature is especially important for the approach presented in this paper.
One can say that flow of processes is not disturbed by any data.
This gives hope to automate the process of generating logical specifications from
workflow-oriented software models which are organized in predefined structures.
The main idea is to associate workflows with temporal logic formulas that describe the dynamic aspects of workflows.
On the other hand, modeling should be limited to a set of predefined workflows and
then models can be developed using only these workflow patterns
as discussed in Section~\ref{sec:analysis-verification}.

If the last rule of Definition~\ref{def:LTL-syntax} is removed,
then the definition of a classical logic formula is received.
These formulas do not contain modal operators ${\cal M}$.
Let us present it more formally.
\begin{definition}
\label{def:CL-syntax}
The \emph{classical logic\emph{, or} point\emph{,} formula}
is a formula which is built using the following rules:
\begin{itemize}
\item if $p \in AP$ then $p$ is a point formula,
\item if $p$ and $q$ are formulas, then $\neg p$, $p \dis q$,
      $p \con q$ are point formulas.
\end{itemize}
\end{definition}
Point formulas allow to describe (logical) circumstances without considering a time flow,
i.e.\ in a point.
Only when they are preceded by a temporal operator
(e.g.\ Algorithm~\ref{alg:generating-specification} or
proof in Theorem~\ref{th:wrk-properties}),
then they are considered in the time context.

Every workflow is linked to logical formulas,
both temporal and classical ones.
Temporal logic formulas enable describing the internal properties of workflows.
Classical logic formulas enable describing workflows from the outside.
These aspects are discussed in more detail below.
\begin{definition}
\label{def:workflow-set}
The \emph{workflow set} of formulas what is denoted $wrf(a_{1}, \ldots, a_{n})$, or simply $wrf()$,
over atomic formulas $a_{1}, \ldots, a_{n}$,
is a set of formulas $f_{en}, f_{ex}, f_{1}, ..., f_{m}$
such that all formulas are syntactically correct,
and $f_{en}$ and $f_{ex}$ are point formulas,
and $f_{1}, ..., f_{m}$ are temporal logic formulas,
i.e.\ $wrf()=\{ f_{en}, f_{ex}, f_{1}, \ldots, f_{m} \}$.
\end{definition}
Formulas $a_{1}, \ldots, a_{n}$ are arguments of a workflow
constituting, informally speaking, its input,
i.e.\ these atomic formulas are used to built both point and temporal formulas of a workflow.
Workflow sets are formed in such a way that
the first two formulas are classical logic ones and further formulas are LTL ones.
The interpretation of such an organization is the following:
\begin{enumerate}
  \item classical logic formulas (Def.~\ref{def:CL-syntax})
        describe (logical) entry or exit points called
        \emph{entry formula} $f_{en}$ or \emph{exit formula} $f_{ex}$ of a workflow,
        i.e.\ they enable representation of a workflow considered as a whole,
        in other words, describing the logical circumstances of, respectively,
        the start and the termination of the whole workflow execution,
        or, they show which activities of a workflow are executed first or last,
        respectively,
        c.f.\ the predefined workflow property set $P$ given in Section~\ref{sec:analysis-verification}.
        Thus, these formulas should not be confused with the well-known
        precondition or postcondition, respectively.
        Let $wrf().f_{en}$ and $wrf().f_{ex}$ are
        entry and exit formulas, respectively, from a workflow set $wrf()$,
        if it does not lead to ambiguity,
        then formulas are written shortly $f_{en}$ and $f_{ex}$;
  \item temporal logic formulas (Def.~\ref{def:LTL-syntax})
        describe the internal behavior of the workflow $f_{1}, ..., f_{m}$,
        showing dynamic aspects of a workflow pattern.
        Every property can be characterized using
        a liveness property and a safety property,
        c.f.~\cite{Alpern-Schneider-1985},
        thus, the aim is to obtain a decomposition
        expressed in terms of temporal logic formulas.
\end{enumerate}
Summing up,
point formulas allow consideration of a workflow as a whole, i.e.\ from the outside point of view,
while temporal formulas show the internal behavior of a workflow.

Some restrictions on atomic formulas $a_{1}, \ldots, a_{n}$ of
the workflow set $wrf()$ in Definition~\ref{def:workflow-set},
due to the partial order, are introduced.
\begin{definition}
\label{def:partial-order}
The set of atomic formulas is divided into three subsets which are pairwise disjointed
and the following rules must be valid:
\begin{enumerate}
\item the first subset which contains at least one element consists of \emph{entry arguments},
      and all of these arguments, and no others, form the $f_{en}$ formula,
\item the second subset which may be empty consists of \emph{ordinary arguments},
\item the third subset which contains at least one element consists of \emph{exit arguments},
      and all these arguments, and no others, form the $f_{ex}$ formula.
\end{enumerate}
\end{definition}

\begin{example}{}
\label{example-workflow-patterns}
Let us discuss some examples of workflow sets for hypothetical workflow patterns:
$W1(a,b)=\{a, b, a \imp\som b, \alw\neg (a \con b)\}$,
$W2(a,b,c)=\{a, b \dis c, a \imp \som b \con \som c, \alw\neg (a \con (b \dis c)) \}$, and
$W3(a,b,c,d)=\{a \dis b, d, a \imp\som c, b \imp\som c, \alw\som c \imp \alw\som d, \alw\neg ((a \dis b) \con (c \dis d)) \}$.
In the case of $W1$ and $W2$ the $a$ proposition is a (logical) starting point for
the whole workflow, i.e.\ it means that when $a$ is satisfied then workflows is started.
$W1$ probably refers to a workflow for a sequence of two tasks $a \imp\som b$ (liveness) and
therefore it is also not possible (safety) that these two task are satisfied
simultaneously $\alw\neg (a \con b)$.
$W2$ probably shows a parallel split of two task,
and therefore the $b \dis c$ formula describes that
when the workflow ends then $b$ or $c$ are satisfied.
In the case of $W3$, the disjunction $a \dis b$ is a (logical) starting point.
The $d$ task is always the last activity of the workflow.
The set of formulas for $W3$ is a more complex and interesting case.
It describes a reactive and fair service (liveness) $\alw\som c \imp \alw\som d$,
i.e.\ when $c$ is satisfied then always follow $d$.
The service is ready to work after a initiation of the whole workflow ($a \dis b$),
and after starting a service (liveness) $a \imp\som c$ or $b \imp\som c$.
It is mandatory to ensure safety of the workflow,
i.e.\ the start formulas and service formulas cannot be satisfied at the same time
$\alw\neg ((a \dis b) \con (c \dis d))$.
\end{example}
\begin{corollary}{}
\label{th:wrk-properties}
The Definition of the workflow set $wrf()$, and further remarks, lead to the following valid statements:
\begin{itemize}
\item none of the ordinary arguments of a workflow set are included either in the $f_{en}$  or the $f_{ex}$ formula;
\item every workflow contains, and its logical formulas describe, the structure that consists of at least two activities (or tasks).
\end{itemize}
\end{corollary}
\begin{proof}
The proof is relatively simple and for example the second statements follows from the fact that
the entire set of atomic formulas $a_{1}, \ldots, a_{n}$ as arguments for a workflow set must contain
at least two arguments which constitute activities (task).
\end{proof}

The whole software model comprising workflows can be quite complex including nesting workflows
and this is why there is a need to define a symbolic notation
which enables to represent any potentially complex structure.
\begin{definition}
\label{def:logical-expression}
The \emph{workflow expression} $W$ is a structure built using the following rules:
\begin{itemize}
\item every workflow set $wrf(a_{1}, ..., a_{i}, ..., a_{n})$,
      where every $a_{i}$ is an atomic formula,
      is a workflow expression,
\item every $wrf(A_{1}, ..., A_{i}, ..., A_{m})$, where every $A_{i}$ is either
      \begin{itemize}
      \item an atomic formula $a_{k}$, where $k>0$, or
      \item a worflow set $wrf(a_{j})$, where $j>0$ and every $a_{j}$ is an atomic formula, or 
      \item a workflow expression $wrf(A_{j})$, where $j>0$
      \end{itemize}
      is also a workflow expression.
\end{itemize}
\end{definition}
Every $a_{i}$ (small letters) represents only atomic formulas.
Every $A_{i}$ (capital letters) represents either atomic formulas or workflows.
These rules allow to define an arbitrary complex workflow expression.
Due to the partial order relation described above and Corollary~\ref{th:wrk-properties},
it should be noted that, in regards to the workflow expression,
there is a similar valid restriction on the number of arguments,
i.e.\ there are at least two arguments for every workflow expression,
what is informally shown through the way of indexing for a workflow expression which takes values $i,j=1,2, ...$

The notion of aggregated entry/exit formulas is introduced which is a result
of nested and complex workflows, as well as the need to transfer,
informally speaking,
the logical signal to all start/termination points of a nested workflow.
\begin{definition}
\label{def:aggregated-conditions}
Let $w^{c}$ for a workflow expression $w$ with the upper index $c=e$ (or $x$, respectively) be
the \emph{aggregated entry formula} (or the \emph{aggregated exit formula}, respectively)
when the aggregated formula is calculated using the following (recursive) rules:
\begin{enumerate}
\item if there is no workflow itself in the place of any atomic formula/argument which syntactically belongs to
      the $f_{en}$ formula (or the $f_{ex}$ formula, respectively) $w$,
      then $w^{e}$ is equal to $f_{en}$ ($w^{x}$ is equal to $f_{ex}$, respectively),
\item if there is a workflow, say $t()$, in a place of any atomic argument, say $r$, which syntactically belongs to
      the $f_{en}$ formula (or the $f_{ex}$ formula, respectively) of $w$,
      then $r$ is replaced by $t^{e}$ (or $t^{x}$, respectively) for every such case.
\end{enumerate}
\end{definition}
These rules allow to define aggregated point formulas for an arbitrary complex workflow expression.

\begin{example}{}
\label{example-predefined-workflow}
Let us supplement Definitions~\ref{def:logical-expression} and~\ref{def:aggregated-conditions} by some examples.
Let $\Sigma$ is a \emph{predefined workflow set}, e.g.\
\begin{eqnarray}
  \Sigma=\{ Seq, Concur, Branch, Loop \} \label{for:sample-prefined-workflow}
\end{eqnarray}
properties of which might be described and stored in the $P$ set,
c.f.\ Fig.~\ref{fig:predefined-P} and~\ref{fig:predefined-P-Arbitrary-Cycles},
modeling sequence, concurrency, branching and iteration, respectively.
However, they are defined here in a different (simpler) way comparing the $P$ set,
i.e.\ through direct introduction of all necessary formulas.
Thus,
$Seq(a,b)=\{a, b, a \imp\som b, \alw\neg (a \con b)\}$,
$Concur(a,b,c)=\{a, b \dis c, a \imp\som b \con \som c,\alw\neg(a \con(b \dis c))\}$,
$Branch(a,b,c)=\{a, b \dis c, a \imp (\som b \con \neg\som c) \dis (\neg\som b \con \som c), \alw\neg(b \con c)\}$, and
$Loop(a,b,c,d)=\{a,d, a \imp (\som b \con \neg\som d) \dis (\neg\som b \con \som d),
b \imp\som c,
c \imp (\som b \con \neg\som d) \dis (\neg\som b \con \som d) \}$.
The meaning of $Seq$ seems obvious.
The $Concur$ and $Branch$ workflows model concurrency and branching,
respectively,
for two activities $b$ and $c$, which are preceded by another activity $a$.
The $Loop$ workflow models a while-cycle case
that has exactly one input activity $a$ and exactly one output activity $d$,
which are located before and after, respectively, the main loop.
The sequence of two activities $b$ and $c$ constitutes the entire body of a loop,
where $b$ is a main instruction of the body,
and $c$ is an incrementation for the body.
Formal definitions in terms of temporal logic formulas for these patterns are proposed above.
If it is necessary to model concurrency and branching without a preceding activity,
then it can be obtained using provided patterns $Concur$ or $Branch$,
and assuming that the preceding activity $a$ may be, informally,
the \emph{null task}, that is the execution of which consumes zero time.
\end{example}

Workflow expressions may represent an arbitrary structure and an example of this is
$Seq(a,Seq(Concur(b,c,d),Branch(e,f,g)))$
meaning of which is intuitive,
i.e.\ it might shows the sequence that leads to
another sequence of a concurrent execution of some activities and then the branch by selecting an activity.
\begin{example}{}
\label{example-aggregated-formula}
Examples of aggregated formulas are given as follows.
For $w=Seq(a,b)$ formulas are $w^{e}=a$ and $w^{x}=b$ (step~1).
For $w=Concur(a,b,Seq(c,d))$ formulas are $w^{e}=a$ (step~1) and
$w^{x}=b \dis d$ (step~1 gives ``$b \dis$'' and step~2 gives $Seq^{x}=$``$d$'' which is aggregated to ``$b \dis d$'').
For $w=Concur(a,b,Concur(c,d,e))$ the formula
$w^{x}=b \dis (d \dis e)$ (step~1 gives ``$b \dis$'' and step~2 gives $Concur^{x}=$``$d \dis e$'' which is aggregated to ``$b \dis (d \dis e)$'').
For $w=Concur(a,Concur(b,c,d),Concur(e,f,g))$ the formula $w^{x}=(c \dis d) \dis (f \dis g)$
(step~2 gives $Concur^{x}=c \dis d$ and $f \dis g$,
and after aggregation ``$(c \dis d) \dis (f \dis g)$'' is obtained).
\end{example}

An important property of workflow expressions is their internal and nested structure.
Parentheses are the best illustration for it.
Suppose that all instances of ``$wrf($'' and ``$,wrf($'',
where ``$wrf$'' is a symbol of an arbitrary workflow,
were substituted by ``$($''.
Then, for example, the above workflow expression leads to the parenthesis structure $(a(b,c,d)(e,f,g))$.
This in turn leads to the following Theorem.
\begin{theorem}{}
\label{theorem:paranthesis}
For any workflow expression and for any two workflow patterns $wrf_{i}()$ and $wrf_{j}()$,
where $i \neq j$,
only one of the following three situations holds:
\begin{enumerate}
\item $wrf_{i}()$ and $wrf_{j}()$ are completely disjointed;
\item $wrf_{i}()$ is completely contained in $wrf_{j}()$;
\item $wrf_{j}()$ is completely contained in $wrf_{i}()$.
\end{enumerate}
\end{theorem}
\begin{proof}
Firstly, let us note that Definition~\ref{def:logical-expression} is recursive.
For the first case of the Definition,
a simple pattern with atomic formulas is considered and it is consistent with the Theorem in an obvious way.
For the second case, two subcases are considered.
For the first subcase,
if the argument is an atomic formula then it is clear that no parentheses are introduced.
For the second subcase,
recursive application of the rule introduces a new pattern with correctly paired parenthesis
and this subcase guarantees the complete contain (nesting) of patterns.
If correctly paired patterns are introduced/substituted in place of different arguments of a pattern
then it guarantees the disjointedness of the paired patterns.
\end{proof}

\section{Generating specifications}
\label{sec:generating-specifications}

The process of generating logical specifications is described below.
Informally speaking, the \emph{logical specification} is a counterpart of a formal generalization
understood as an act of taking some facts and making a broader statements,
i.e.\ the formal derivation of a general statement from a particular one.
In this work the logical specification is expressed as a set of temporal logic formulas.
These formulas are generated from workflow expressions using predefined workflows as logical primitives.
Let us define it formally and then present an algorithm.
\begin{definition}
\label{def:logical-specification}
The \emph{logical specification} $L$ is a set of temporal logic formulas derived from
a workflow expression~$W$ and predefined set~$P$
using the algorithm $\Pi$,
i.e.\ $L(W) = \{f_{i} : i>0 \con f_{i} \in \Pi(W,P)\}$,
where $f_{i}$ is a LTL formula.
\end{definition}

Generating logical specifications is not a simple summation of
predefined formula collections resulting from patterns used in a workflow expression.
The generation algorithm~$\Pi$ is given as Algorithm~\ref{alg:generating-specification}.
The generation process has two inputs.
The first one is a workflow expression $W$ which is a kind of variable,
i.e.\ it varies for every workflow model.
The second one is a workflow property set~$P$ which is a kind of constant as it is predefined and fixed
containing definitions of workflows in terms of temporal logic formulas.
The more detailed information about this set including its examples is in Section~\ref{sec:analysis-verification}.
The output of the generation algorithm is a logical specification understood as
a set of temporal logic formulas.
Let $wrf()^{T}$ represents a set of all temporal formulas extracted from a workflow set (i.e.\ without point formulas).
\begin{algorithm}[htb]
\caption{Generating logical specifications ($\Pi$)}
\label{alg:generating-specification}
{\normalsize
\begin{algorithmic}[1]
\algrenewcommand\algorithmicrequire{\textbf{Input:}}
\algrenewcommand\algorithmicensure{\textbf{Output:}}
\Require Logical expression $W_{L}$ (non-empty), predefined set $P$ (non-empty)
\Ensure Logical specification $L$
\State $L:=\emptyset$ \Comment{initiating specification}\label{alg:2:atomic-ini}
\For{every workflow $wrf()$ of $W_{L}$ from left to right}\label{alg:2:atomic-for}
\If{all arguments of $wrf()$ are atomic}\label{alg:2:atomic-s}
\State{$L := L \cup wrf()^{T}$}
\EndIf\label{alg:2:atomic-e}
\If{any argument of $wrf()$ is a workflow itself}\label{alg:2:non-atomic-s}
\State{for every such an argument, say $r()$, substitute}
\State{disjunction of its aggregated entry and exit}
\State{formulas in all places where the argument}
\State{occurs in the $wrf()$ temporal formulas, i.e.}
\State{$L := L \cup ((wrf()^{T}) \leftarrow \mbox{``$r()^{e} \dis r()^{x}$''})$}\label{alg:2:non-atomic-dis}
\EndIf\label{alg:2:non-atomic-e}
\EndFor
\end{algorithmic}
}
\end{algorithm}
The Algorithm refers to similar ideas in the works~\cite{Klimek-2013-sefm,Klimek-Faber-Kisiel-Dorohinicki-2013-fedcsis},
however, the case considered here is a more general and not focused on specific patterns.
All workflows of the workflow expression are processed one by one and the Algorithm always halts.
All parentheses are paired.
Let $p4(h,p2(d,p1(a,b,c),e),p3(f,g))$ is a hypothetical workflow expression,
where $p1$, $p2$, $p3$, and $p4$ are workflow patterns.
Pattern $p3$ has two arguments, and other patterns have three arguments.
Considering the loop in the line~\ref{alg:2:atomic-for} of Algorithm~\ref{alg:generating-specification},
the processing order of patterns is the following: $p4$, $p2$, $p1$, and $p3$,
where $p4$ and $p2$ are processed in lines \ref{alg:2:non-atomic-s}--\ref{alg:2:non-atomic-e},
and $p1$ and $p3$ are processed in lines \ref{alg:2:atomic-s}--\ref{alg:2:atomic-e}.

\begin{example}{}
\label{example-algorithm}
Considering the predefined workflow set given by Formula~\ref{for:sample-prefined-workflow},
and its definitions of workflows,
let us supplement Algorithm~\ref{alg:generating-specification} by some examples.
The example for lines \ref{alg:2:atomic-s}--\ref{alg:2:atomic-e}:
$Seq(a,b)$, gives $L=\{ a \imp\som b, \alw\neg (a \con b) \}$
and $Branch(a,b,c)$ gives $L=\{ a \imp (\som b \con \neg\som c) \dis (\neg\som b \con \som c), \alw\neg (b \con c) \}$.
The example for lines \ref{alg:2:non-atomic-s}--\ref{alg:2:non-atomic-e}:
$Concur(Seq(a,b),c,d)$ leads to
$L = \{ (a \dis b) \imp \som c \con \som d, \alw\neg((a \dis b) \con(c \dis d)) \} \cup \{a \imp\som b, \alw\neg (a \con b) \}$.
Other examples are shown in Section~\ref{sec:analysis-verification}.
\end{example}

The Algorithm comprises two main parts.
In the first part, lines \ref{alg:2:atomic-s}--\ref{alg:2:atomic-e},
logical specifications are rewritten from a predefined set,
c.f.\ Fig~\ref{fig:predefined-P} and~\ref{fig:predefined-P-Arbitrary-Cycles},
without any modification and summed with the resulting specification.
In the second part, lines \ref{alg:2:non-atomic-s}--\ref{alg:2:non-atomic-e},
the workflow $f_{en}$ and $f_{ex}$ formulas are taken into account
since they allow consideration of the nested workflow as a whole,
i.e.\ without analyzing its internal behavior
which is itself and separately taken into account in the first part.
Consideration of both $f_{en}$ and $f_{ex}$ seems a bit redundant for a single workflow
but on the other hand, informally speaking,
these two formulas have equal rights to represent a workflow,
and the line~\ref{alg:2:non-atomic-dis} contains their disjunction which is substituted,
and then modified temporal formulas are summed with the resulting specification.

The Algorithm allows to automate the process of generating logical specifications.
Logical expressions are translated into logical specifications which are expressed in
terms of temporal logic formulas.
Logical expressions can be arbitrarily complex and nested.
Moreover,
the list of predefined patterns can be arbitrarily,
that is in any way and at any time,
extended by new patterns.
The only requirement is to define behaviour,
c.f.\ Fig~\ref{fig:predefined-P} and~\ref{fig:predefined-P-Arbitrary-Cycles},
for new patterns in terms of temporal logic prior to their first use.
Thus,
the general idea that logical patterns are once defined and then widely used is satisfied.

The completeness problem is a fundamental issue for logical systems
and constitutes their key requirement in many fields.
Informally speaking, completeness is some opposition to the fragmentation.
In other words,
completeness means having all elements and lacking nothing while
fragmentation means not having all elements and lacking something.
Generally speaking, an object, or a set of objects, is complete if nothing more needs to be added to it.
In formal logic systems, \emph{completeness} means that if a formula is valid, it can be proven~\cite[p.~128]{Gries-Schneider-1993}.
In algorithms, it refers to the ability of finding a solution if one exists.

This paper discusses both predefined logical specifications and
the Algorithm for generating logical specifications using predefined specifications.
This requires an integrated perspective towards completeness by considering two aspects:
\begin{enumerate}
\item completeness of possessed logical specifications,
      that is contained in a predefined set of patterns,
      c.f.\ Fig.~\ref{fig:predefined-P} and~\ref{fig:predefined-P-Arbitrary-Cycles}, and
\item completeness of the generation algorithm,
      i.e.\ Algorithm~\ref{alg:generating-specification}.
\end{enumerate}
Firstly, the completeness of the predefined set~$P$ is considered.
The set consists of logical specifications that refer to particular patterns,
or, in other words, every pattern is defined in terms of temporal logic formulas.
These specifications should be examined, one by one, for compliance with the relevant logical properties.
However,
as it has already been said in Section~\ref{sec:motivation-contribution} (motivation),
logical patterns are predefined by a logician or a person with good skills in logic
for further use by an ordinary analyst or a developer.
This leads to the conclusion that the logician is responsible for proving the correctness and logical properties of predefined specifications,
and some decision procedures,
c.f.\ Corollary~\ref{th:decision-procedures},
might be helpful for this process.

Predefined logical specifications constitute an input for the generation Algorithm~\ref{alg:generating-specification}.
Thus, it is reasonable to question whether the algorithm preserves the completeness when generating the resulting logical specification,
i.e.\ obtained as an output of the Algorithm.
\begin{definition}
\label{def:relatively-completene}
The Algorithm of generating logical specification is \emph{relatively complete}
if it preserves completeness of the generated logical specification,
or, in other words,
if it does not introduce itself incompleteness to the output logical specifications
with respect to predefined input specifications.
\end{definition}
\begin{theorem}{}
Supposing that predefined workflow set is non-empty, and every pattern of the $P$ set is non-empty,
and every two patterns have disjointed sets of atomic formulas,
and the workflow expression $W$ is non-empty,
then the logical specification obtained for Algorithm~\ref{alg:generating-specification} is relatively complete.
\end{theorem}
\begin{proof}
Let us note that due to the parenthesis Theorem~\ref{theorem:paranthesis},
patterns are nested entirely/completely,
i.e.\ it is not possible to obtain a partial nesting
that might provide an undesirable crossing of patterns.
Furthermore, every two patterns contain disjointed sets of atomic formulas.
Every system can be described in terms of safety and liveness properties/formulas~\cite{Alpern-Schneider-1985}.
If a predefined logical specification is complete,
then incompleteness can not be introduced while generating the output logical specification
when using liveness formulas.
The most general form for liveness is formula $P \imp\som Q$.
Let us consider two cases for the Algorithm.
\begin{enumerate}
  \item \textit{Case for lines \ref{alg:2:atomic-s}--\ref{alg:2:atomic-e}.}
        Specifications are only rewritten from a predefined set,
        c.f.\ Fig.~\ref{fig:predefined-P} and Fig.~\ref{fig:predefined-P-Arbitrary-Cycles},
        then if the input specification is complete, than the completeness property is preserved.
  \item \textit{Case for lines \ref{alg:2:non-atomic-s}--\ref{alg:2:non-atomic-e}.}
        The entry and exit formulas are considered. They are generalization for
        a nested pattern and allows to bypass/skip its internal behaviour.
        They enable considering both the beginning and the end of a workflow.
        Let us note that for any workflow pattern $w()$,
        due to Corollary~\ref{th:wrk-properties},
        there is always satisfied $\alw\neg(w().f_{en} \con w().f_{ex})$.
        On the other hand, there is also valid $\alw(w().f_{en} \imp\som w().f_{ex})$.
        However, due to the nature of entry and exit points,
        they are both either satisfied or not satisfied,
        mapping a kind of logical propagation,
        that leads to the third, and additional, formula $\alw(\neg w().f_{en} \imp \neg\som w().f_{ex})$.

        Completeness refers to the reachability of all formulas and properties of a logical specification.
        Let us consider the sequence of two workflows $Seq(g(),h())$
        for the predefined set expressed by Formula~\ref{for:sample-prefined-workflow} and further definitions.
        Let $g().f_{en} \equiv g_{e}$, $g().f_{ex} \equiv g_{x}$, $h().f_{en} \equiv h_{e}$, and $h().f_{ex} \equiv h_{x}$.
        Let us return to the three formulas introduced above, and after considering them in the context of workflows $g()$ and $h()$,
        they are gathered as premises.
        Now,
        due to the ordinary liveness formula $a \imp\som b$,
        e.g.\ definition of $Seq$ for Formula~\ref{for:sample-prefined-workflow},
        where $a$ refers to $g()$ and $b$ refers to $h()$ for the mentioned $Seq(g(),h())$,
        and the substitution in the~\ref{alg:2:non-atomic-dis}-th line of Algorithm~\ref{alg:generating-specification},
        the following formula that is added to the premises set is obtained:
        $\alw((g_{e} \dis g_{x}) \imp \som(h_{e} \dis h_{x})))$.
        Formula $\alw(g_{x} \imp \som h_{e})$ is a requirement that expresses the demand
        to pass from one exit point directly to the next entry point that allows to cover
        all properties/formulas from the beginning of the next workflow.
        Gathering all premises and the demand, the resulting formula is
        \begin{eqnarray}
           (\alw(g_{e} \imp \som g_{x}) \con \alw(\neg g_{e} \imp \neg\som g_{x}) \con\nonumber\\
           \alw\neg(g_{e} \con g_{x}) \con \alw(h_{e} \imp \som h_{x}) \con\nonumber\\
           \alw(\neg h_{e} \imp \neg\som h_{x}) \con \alw\neg(h_{e} \con h_{x}) \con \nonumber\\
           \alw((g_{e} \dis g_{x}) \imp \som(h_{e} \dis h_{x})))\nonumber\\
           \imp \alw(g_{x} \imp \som h_{e}) \label{for:compelte}
        \end{eqnarray}
        While analyzing the above Formula using the semantic tableaux method,
        the obtained truth tree, similar to the small tree from Fig.~\ref{fig:deduction-tree},
        contains many hundreds of nodes,
        and is closed which means that Formula~\ref{for:compelte} is always satisfied (tautology).
\end{enumerate}
Considering both cases is sufficient for the entire Algorithm.
\end{proof}

\section{Models analysis and verification}
\label{sec:analysis-verification}

The method of formal verification of business models is discussed in this Section.
The method follows from the approach provided in this work.
Firstly, business systems are modeled using predefined workflow patterns,
i.e.\ processes associated with logical patterns.
In other words workflow patterns are predefined in terms of temporal logic formulas,
and then logical specifications are automatically generated using the proposed Algorithm~\ref{alg:generating-specification}.
The introduced deduction-based verification system allows to perform verification of business models in the formal way.

Workflow patterns are crucial for the approach introduced in this work as they lead to
the automation of the logical specifications generation process.
Informally speaking, a pattern is a distinctive formation created and used as an archetype.
Creating and using patterns promotes software reuse which is always a kind of \emph{id{\'e}e fixe} in software engineering.
Riehle and Zullighoven in their work~\cite{Riehle-Zullighoven-1996} described \emph{patterns} as
``the abstraction from a concrete form which keeps recurring in specific non-arbitrary contexts''.
Patterns might constitute a kind of primitives which enable mapping of the workflow patterns to logical specifications.
\emph{Business Process Modeling Notation} (BPMN) is a standard and dominant graphical notation,
e.g.~\cite{BPMN-OMG-2011},
for the modeling of business processes.
The primary goal of BPMN is to provide a notation that is
understandable by all business users, from business analysts to technical developers,
and finally, to business people who will manage and monitor these processes.
An important part of BPMN are 21 patterns which are introduced
in the work~\cite{Aalst-etal-2003}.
Gradually building in complexity, process patterns were broken down into six categories,
and the Basic Control Flow Patterns category is considered in this work.
The proposed method of the automatic extraction of logical specifications is
based on the assumption that the whole business model is built using
only the well-known workflow patterns of BPMN.
This assumption is fundamental to the consideration of the work and
is not a restriction since it enables receiving correct
and well-composed business models.

\begin{figure}[htb]
{\footnotesize
\begin{minipage}{.4\linewidth}
\begin{verbatim}
                        /* version 25.10.2013
/* Basic Control Patterns
Sequence(f1,f2):
f1
f2
[](f1 => <>f2) / [](~f1 => ~<>f2)
[]~(f1 & f2)
ParallelSplit(f1,f2,f3):
f1
f2 | f3
[](f1 => <>f2 & <>f3) / [](~f1 => ~<>f2 & ~<>f3)
[]~(f1&(f2|f3))
Synchronization(f1,f2,f3):
f1 | f2
f3
[](f1 & f2 => <>f3) / []( ~(f1 & f2) => ~<>f3)
[]~((f1|f2)&f3)
ExclusiveChoice(f1,f2,f3):
f1
f2 | f3
[](f1 => (<>f2 & ~<>f3)|(~<>f2 & <>f3))
[](~f1 => ~<>f2 & ~<>f3)
[]~(f1&(f2|f3)) / []~(f2 & f3)
SimpleMerge(f1,f2,f3):
f1 | f2
f3
[](f1|f2 => <>f3) / [](~(f1|f2) => ~<>f3)
[]~(f1|f2) / []~((f1|f2)&f3)

/* ..... [other] Business Patterns
\end{verbatim}
\end{minipage}
}
\caption{A sample predefined set $P$}
\label{fig:predefined-P}
\end{figure}
Let the predefined workflow set of patterns be
$\Sigma=\{Sequence, ParallelSplit,$ $Synchronization, ExclusiveChoice, SimpleMerge\}$.
This set might be extended using other patterns described in the work~\cite{Aalst-etal-2003}.
Definitions of all potentially used workflow patterns are expressed in terms of temporal logic and
stored in the set $P$, which is predefined and fixed.
It is assumed that the defining process is performed by a person with good skills in logic.
The process should contains considerations and proofs of the logical properties for every pattern.
Furthermore, the defining process is performed once, and then logical primitives can be widely used.
The example of such a predefined workflow set $P$ is shown in Figure~\ref{fig:predefined-P}.
The way to define formally the individual workflow patterns,
the type of used formulas, is itself an interesting problem.
However, it is not discussed here exactly,
and should be the subject of research for all patterns in a separate work,
c.f.\ remarks in the last Section~\ref{sec:conclusions},
where the syntax of the presentation language shown in the Figure is expected to be defined formally.
Now, it is presented informally in the following way.
Most elements of the $P$~set,
i.e.\ two temporal logic operators, classical logic operators, are not in doubt in understanding.
The slash allows to place more than one formula in a single line.
$f_{1}$, $f_{2}$ etc.\ are atomic formulas
and constitute a kind of formal arguments for a pattern.
Every pattern has two point formulas which are located at the beginning of the set
describing the start and final, respectively,
logical conditions/circumstances of the execution of a pattern.
The content of the $P$~set is shown as a plain ASCII text to
illustrate its participation in the real processing,
c.f.\ Fig.~\ref{fig:deduction-system}.
Though, the above set contains a relatively small number of patterns,
justification for this is only to present a general idea for the pattern-oriented generation of logical specifications,
and there is no difficulty with defining a set of workflow formulas for
any other process patterns, as well as for the 21 patterns mentioned above~\cite{Aalst-etal-2003,White-2004}.

\begin{figure}[htb]
{\footnotesize
\begin{minipage}{.4\linewidth}
\begin{verbatim}

ArbitraryCycles(Alfa,Beta,Chi,A,B,C,D,F,E,G):
Alfa
E | G
/* first loop (Alfa)
[](x(Alfa) & c(Alfa) => <>B & ~<>A)
[](x(Alfa) & ~c(Alfa) => <>A & ~<>B)
[](~x(Alfa) => ~<>A & ~<>B)
[]~(x(Alfa)&(A|B)) / []~(A|B|C)
[](A => <>C) / [](~A => ~<>C)
[](B | C => <>D) / [](~(B | C) => ~<>D)
[]~((B|C)&D)
/* second loop (Beta)
[](D => <>x(Beta)) / [](~D => ~<>x(Beta))
[]~(D&x(Beta))
[](x(Beta) & c(Beta) => <>E & ~<>F)
[](x(Beta) & ~c(Beta) => ~<>E & <>F)
[](~x(Beta) => ~<>E & ~<>F)
[]~(x(Beta)&(E|F)) / []~(E|F)
/* towards outside (Chi,G)
[](F => <>x(Chi)) / [](~F => ~<>x(Chi))
[]~(F&x(Chi))
[](x(Chi) & c(Chi) => <>G & ~<>C)
[](x(Chi) & ~c(Chi) => ~<>G & <>C)
[]~(x(Chi)&(G|C)) / []~(G|C)
\end{verbatim}
\end{minipage}
}
\caption{The Arbitrary Cycles pattern for a predefined set $P$}
\label{fig:predefined-P-Arbitrary-Cycles}
\end{figure}
Although formal definitions for all patterns exceed the size and goal of this paper,
the predefined set of workflows~$P$ is to be extended by the ``ArbitraryCycles'' pattern,
c.f.~\cite[pages 11--12]{White-2004} and Fig.~\ref{fig:predefined-P-Arbitrary-Cycles},
which is perhaps the most complex process pattern, and
$\Sigma := \Sigma \cup \{ArbitraryCycles\}$.
The pattern represents cycles that have more than one entry or exit points.
There are no special restrictions on the used types of loops.
A new notation linked with tested loop conditions is introduced.
If $exp$ is a condition (logical expression) to be tested,
which is associated with a certain activity,
then $c(exp)$ means that the logical expression~$exp$ is evaluated and is true.
$x(exp)$ means that the activity associated with the expression $exp$ is satisfied,
i.e.\ the activity is executed (from the rising/positive edge to the falling/negative edge).
$exp$ can be evaluated only when $x(exp)$ is satisfied,
and the following sentence is valid: $x(exp) \con (c(exp) \dis \neg c(exp))$,
otherwise, when $\neg x(exp)$, the value of the $c(exp)$ expression is undefined.
The example of an extended part of the $P$ set is shown in Figure~\ref{fig:predefined-P-Arbitrary-Cycles}.
$Alfa$, $Beta$, $Chi$, $A$, $B$, $C$, etc.\ are formal arguments for the workflow pattern,
where the first three arguments refer with some conditions.
The workflow has one entry argument, two exit arguments, and six ordinary arguments.
Temporal formulas of the workflow set describe both safety and
liveness properties for the pattern.

\begin{example}{}
\label{example-main}
Let us consider a simple yet illustrative example to present the approach of the work.
The example is somewhat abstract but the main purpose is
to demonstrate the main idea which is the deployment of predefined patterns for modeling and generating logical specifications,
and formal verification of business models.
Suppose workflow expression~$W$ is
\begin{eqnarray}
    Sequence(ExclusiveChoice(Sequence(a,b),\nonumber\\
    Sequence(c,d),Sequence(ParallelSplit(e,f,g),\nonumber\\
    Synchronization(h,i,j))),SimpleMerge(k,l,m))\nonumber
\end{eqnarray}
The logical specification $L$ is built in the following steps.
At the beginning, the specification is $L=\emptyset$.
The patterns are processed in the following order:
$Sequence$,
$ExclusiveChoice$,
$Sequence$,
$Sequence$,
$Sequence$,
$ParallelSplit$,\\
$Synchronization$, and
$SimpleMerge$.

The following sub-sets are generated:
the first $Sequence$ gives
$L_{1}= \{ \alw(a \dis (d \dis j)) \imp \som ((k \dis l) \dis m),
       \alw(\neg(a \dis (d \dis j)) \imp \neg\som ((k \dis l) \dis m),
       \alw\neg((a \dis (d \dis j)) \con ((k \dis l) \dis m)) \}$,
$ExclusiveChoice$ gives
$L_{2}= \{ \alw((a \dis b) \imp (\som (c \dis d) \con \neg\som (e \dis j)) \dis (\neg\som (c \dis d) \con \som (e \dis j))),
       \alw(\neg (a \dis b) \imp \neg\som (c \dis d) \con \neg\som (e \dis j)),
       \alw\neg((a \dis b) \con ((c \dis d) \dis (e \dis j))),
       \alw\neg((c \dis d) \con (e \dis j)) \}$,
the second $Sequence$ gives
$L_{3}= \{ \alw(a \imp \som b),
       \alw(\neg a \imp \neg\som b),
       \alw\neg(a \con b) \}$,
the third $Sequence$ gives
$L_{4}= \{ \alw(c \imp \som d),
       \alw(\neg c \imp \neg\som d),
       \alw\neg(c \con d) \}$,
the fourth $Sequence$ gives
$L_{5}= \{ \alw((e \dis (f \dis g)) \imp \som ((h \dis i) \dis j)),
       \alw(\neg (e \dis (f \dis g)) \imp \neg\som ((h \dis i) \dis j)),
       \alw\neg((e \dis (f \dis g)) \con ((h \dis i) \dis j)) \}$,
$ParallelSplit$ gives
$L_{6}= \{ \alw(e \imp \som f \con \som g),
       \alw(\neg e \imp \neg\som f \con \neg\som g),
       \alw\neg(e \con (f \dis g)) \}$,
$Synchronization$ gives
$L_{7}= \{ \alw(h \con i \imp \som j),
       \alw( \neg(h \con i) \imp \neg\som j),
       \alw\neg((h \dis i) \con j) \}$, and
$SimpleMerge$ gives
$L_{8}= \{ \alw(k \dis l \imp \som m),
       \alw(\neg(k \dis l) \imp \neg\som m),
       \alw\neg(k \dis l),
       \alw\neg((k \dis l) \con m) \}$.
Thus, the resulting specification is
$L = L_{1} \cup ... \cup L_{8}$
and contains formulas
\begin{eqnarray}L=\{
\alw(a \dis (d \dis j)) \imp \som ((k \dis l) \dis m),\nonumber\\
       \alw(\neg(a \dis (d \dis j)) \imp \neg\som ((k \dis l) \dis m),\nonumber\\
       \alw\neg((a \dis (d \dis j)) \con ((k \dis l) \dis m)),\nonumber\\
       \alw((a \dis b) \imp (\som (c \dis d) \con \neg\som (e \dis j)) \dis\quad\nonumber\\ (\neg\som (c \dis d) \con \som (e \dis j))),\nonumber\\
       \alw(\neg (a \dis b) \imp \neg\som (c \dis d) \con \neg\som (e \dis j)),\nonumber\\
       \alw\neg((a \dis b) \con ((c \dis d) \dis (e \dis j))),\nonumber\\
       \alw\neg((c \dis d) \con (e \dis j)),
\alw(a \imp \som b),\nonumber\\
       \alw(\neg a \imp \neg\som b),
       \alw\neg(a \con b),\nonumber\\
\alw(c \imp \som d),
       \alw(\neg c \imp \neg\som d),
       \alw\neg(c \con d),\nonumber\\
\alw((e \dis (f \dis g)) \imp \som ((h \dis i) \dis j)),\nonumber\\
       \alw(\neg (e \dis (f \dis g)) \imp \neg\som ((h \dis i) \dis j)),\nonumber\\
       \alw\neg((e \dis (f \dis g)) \con ((h \dis i) \dis j)),\nonumber\\
\alw(e \imp \som f \con \som g),
       \alw(\neg e \imp \neg\som f \con \neg\som g),\nonumber\\
       \alw\neg(e \con (f \dis g)),
\alw(h \con i \imp \som j),\nonumber\\
       \alw( \neg(h \con i) \imp \neg\som j),
       \alw\neg((h \dis i) \con j),\nonumber\\
\alw(k \dis l \imp \som m),
       \alw(\neg(k \dis l) \imp \neg\som m),\nonumber\\
       \alw\neg(k \dis l),
       \alw\neg((k \dis l) \con m)
                  \} \label{formula-example-specification}
\end{eqnarray}
The resulting logical specification can be used for the formal verification of a system.
\end{example}

Liveness and safety are standard taxonomy of properties when specifying and verifying systems.
\emph{Liveness} means that the computational process achieves its goals,
i.e.\ something good eventually happens, or,
its counterexample has a prefix extended to infinity.
\emph{Safety} means that the computational process avoids undesirable situations,
i.e.\ something bad never happens, or,
its counterexample has a finite prefix.
The liveness property for the model can be
\begin{eqnarray}
\alw(b \imp\som j) \label{examined-property-liveness}
\end{eqnarray}
what means that always if $b$ is satisfied then sometime in the future the $j$ activity is satisfied.
The safety property for the examined model can be
\begin{eqnarray}
\alw\neg (c \con g) \label{examined-property-safety}
\end{eqnarray}
what means that
it never occurs that $c$ and $g$ are satisfied in the same time.

The whole formula to be analyzed using the semantic tableaux method for
the property expressed by Formula~\ref{examined-property-liveness} is
\begin{eqnarray}
(\alw(a \dis (d \dis j)) \imp \som ((k \dis l) \dis m) \con ... \con\nonumber\\
\alw\neg((k \dis l) \con m)) \imp
               \alw(b \imp\som j) \label{formula-example-property-liveness}
\end{eqnarray}
Formula~\ref{formula-example-specification} represents the output of
the \psframebox[boxsep=true,shadow=false]{\scriptsize G} component in Fig.~\ref{fig:deduction-system}.
Formula~\ref{formula-example-property-liveness} provides a combined input for
the \psframebox[boxsep=true,shadow=false]{\scriptsize T} component in Fig.~\ref{fig:deduction-system}.
When considering the property expressed by Formula~\ref{examined-property-safety},
then the whole formula is constructed in a similar way as
\begin{eqnarray}
(\alw(a \dis (d \dis j)) \imp \som ((k \dis l) \dis m) \con ... \con \nonumber\\
\alw\neg((k \dis l) \con m)) \imp
(\alw\neg (c \con g))\quad \label{formula-example-property-safety}
\end{eqnarray}
The full reasoning tree for both cases contains hundreds of nodes.
Formulas are valid and the examined properties are satisfied in
the considered model.

The prover is an important component of the architecture for the deduction-based system shown in Fig.~\ref{fig:deduction-system}.
It enables automate the inferencing process and formal verification of developed models.
Reasoning engines are more available,
especially in recent years when a number of provers for modal logics became accessible,
c.f.~\cite{Schmidt-2013-provers}.
Selection of an appropriate existing prover,
or building one's own,
constitutes a separate task that exceeds the size and main objectives of this work,
c.f.\ also the concluding remarks in the last Section.

\section{Conclusions}
\label{sec:conclusions}

The method of a pattern-oriented automatic generation of logical specifications for
business models expressed in BPMN is proposed.
Logical specifications are considered as a set of temporal logic formulas and obtaining it
is a crucial aspect in the case of the practical use of the deduction-based formal verification.
The algorithm as a method for an automatic generation of logical specifications from predefined logical patterns/primitives is proposed.
The architecture of a deduction-based system for formal verification of business models is presented.

The method of generating enables a kind of scaling up,
migration from small problems to real-world problems
in this sense that they are having more and more nesting patterns.
This gives hope for practical use in the case of problem of any size.
The proposed approach introduces the concept of logical primitives,
workflow patterns predefined in terms of temporal logic formulas.
They might be once well-defined and could be widely used by an inexperienced user.
The proposed system enables formal verification of business models
using temporal logic and semantic tableaux provers.
The advantage of the method is to provide an innovative concept for
the process verification which might be done for any given business model created
using the BPMN notation.
The aim of this work has been to provide a conceptual theoretical framework
to prepare workable solutions for the deduction-based formal verification of workflow-oriented models.

Future research should extend the results in some directions,
e.g.\ other logical properties of the approach should be explored.
The fundamental issue for the approach is to define formally all workflow patterns~\cite{Aalst-etal-2003}
in terms of temporal logic formulas to provide temporal logic-based semantics for workflows.
The literature review argues that there is a lack of such comprehensive and formal definitions.
Definitions proposed in Section~\ref{sec:analysis-verification},
i.e.\ the predefined set $P$,
might be considered as the beginning of such work.
Another important issue is a detailed analysis of the existing and available provers~\cite{Schmidt-2013-provers}
which could be useful for the approach and applied as a prover component (Fig.~\ref{fig:deduction-system}).
Future works may also include both the implementation of the generation component (Fig.~\ref{fig:deduction-system})
and its own temporal logic prover (Section~\ref{sec:preliminaries}) using the semantic tableaux method.
Implementation works in both of these cases are carried out and relatively advanced.
It should result in a CASE software providing industrial-proof tools,
that is implementing another part of formal methods, hope promising, in industrial practice.


\bibliography{rk-bib-rk,rk-bib-main,rk-bib-ws}


\end{document}